\documentclass[11pt]{article}
\usepackage{fullpage}
\usepackage{lettrine}
{}
\newtheorem{proof}{Proof}{}
{}
\newtheorem{remark}{Remark}{}
\newtheorem{lemma}{Lemma}
\usepackage{authblk}
\usepackage{relsize}
\usepackage{float}
\usepackage{stfloats}
\usepackage{lipsum}
\usepackage{mathtools}
\usepackage{cuted}
\usepackage{pdfpages}
\usepackage[T1]{fontenc}
\usepackage[utf8]{inputenc}
\usepackage{tgbonum}
\graphicspath{ {./img/} }
\usepackage{amsmath,amssymb,amsfonts}
\usepackage[thinc]{esdiff}
\usepackage{textcomp}
\usepackage{xcolor}
\usepackage[linesnumbered,ruled,vlined]{algorithm2e}
\usepackage[noend]{algpseudocode}

\SetCommentSty{mycommfont}
\SetKwInput{KwInput}{Input} 
\SetKwInput{KwRequire}{\normalfont Require}               % Set the Input
\SetKwInput{KwStep}{\normalfont 1}
\SetKwInput{KwEmpty}{}
\SetKwInput{KwStepp}{\normalfont 2}
\SetKwInput{KwSteppp}{\normalfont 3}
\SetKwInput{KwStepppp}{\normalfont 4}
\SetKwInput{KwSteppppp}{\normalfont 5}
\SetKwInput{KwOutput}{\normalfont Output}
\title{{Low-Complexity Angle-Domain MIMO NOMA System with partial channel state information for MmWave Communications}\thanks{This work has been funded by  IMT Atlantique, Lebanese University research fund program and the AZM association.}}
\author[1,2]{Israa Khaled}
\author[1]{Charlotte Langlais}
\author[2]{Ammar El Falou}
\author[2]{Bachar ElHassan}
\author[1]{Michel Jezequel}
\affil[1]{Electronics Department, Institut Mines-Telecom, TelecomBretagne, CNRS UMR 6285 Lab-STICC, \\ CS83818 - 29238 Brest Cedex 3, France}
\affil[2]{ Lebanese University, Faculty of Engineering, Tripoli, Lebanon}
\date{}
\begin{document}
	\maketitle
\begin{abstract}
In millimeter-wave communication, digital beamsteering (DBS), only based on the user direction, is a promising angle-domain multi-antenna technique to mitigate the severe path loss and multi-user interference, with low-complexity and partial channel state information (CSI). In this paper, we design a power-domain non-orthogonal multiple access (NOMA) scheme that enhances the DBS performance trading-off complexity, energy-consumption and capacity performance. In particular, we propose a user-clustering algorithm to pair users, based on a geometric-interference metric, so that  the inter-user interference is reduced. Afterward, based on a fixed inter-cluster power allocation, we derive analytically a sub-optimal intra-cluster power allocation optimization problem to maximize the network throughput. To address the issue of partial CSI, we rewrite the aforementioned optimization problem, by relying only on the user direction. Performance evaluation of the proposed schemes is developed in rural environment, based on the New York University millimeter-wave simulator. The obtained results demonstrate that the proposed low-complexity NOMA-DBS schemes with either full- or partial-CSI achieve significant performance improvement over the classical DBS, in terms of spectral- and energy-efficiencies (up to $26.8\%$ bps/Hz rate gain for 45 users using the proposed scheme with partial CSI).
\end{abstract}

\section{Introduction}
\label{sec:intro}	 
To meet the exponential growth of data traffic, the new generation of cellular systems needs to achieve higher network capacity and better energy efficiency. This is done by exploiting new and innovative technologies such as massive multiple-input multiple-output (m-MIMO) system, non-orthogonal multiple access (NOMA) technique and millimeter-wave (mmWave) bands \cite{busari2017millimeter,araujo2016massive,dai2015non}. These technologies can be integrated together to further enhance the network capacity \cite{zhang2017capacity}. Indeed, the m-MIMO beamforming (BF), which enables high directional array, overcomes the tremendous path loss in mmWave wireless communications where enormous bandwidths are available. In addition, using superposition coding (SC) at the base station (BS) and successive interference cancellation (SIC) at the receiver side, power-domain NOMA (PD-NOMA) effectively improves the connectivity density by multiplexing user equipments (UEs) into transmission power domain \cite{dai2015non}.

Recently, numerous research studies on NOMA-based multiuser MIMO (NOMA-MIMO) system have been initiated to enhance the spectrum efficiency. For instance, system-level simulations are presented in \cite{benjebbovu2013system} and show a clear superiority of NOMA over orthogonal multiple access (OMA) in terms of system throughput. Moreover, the authors in \cite{zeng2017capacity} analytically compares the performance of NOMA-MIMO with OMA-MIMO in terms of the sum-channel- and ergodic-sum-capacities, when multiple users exist in a cluster. The analytical proof also indicates the superiority of NOMA-MIMO. 

The most important challenges of NOMA-MIMO system include overall system overhead, user clustering, power allocation and beamforming techniques. For multicast NOMA-MIMO, an iterative algorithm is proposed in \cite{choi2015minimum} to find the BF vector and the power for each UE that solves the power minimization problem. In \cite{xiao2018joint}, a joint beamforming and power allocation scheme was proposed to maximize the sum-rate of a 2-user mmWave NOMA-MIMO system using an analog BF structure with a phased array. A spectrum and energy efficient mmWave transmission scheme that integrates NOMA with beamspace MIMO was first proposed in \cite{wang2017spectrum}, in order to reduce the hardware complexity and energy consumption by using lens antennas. In \cite{wang2017spectrum}, the authors designed a precoding scheme based on the principle of zero forcing  (ZF) to reduce the inter-cluster interferences. 
%In \cite{zeng2017capacity}, a power allocation was proposed for 2-user NOMA-MIMO without BF.

However, \cite{choi2015minimum,wang2017spectrum,xiao2018joint, zeng2017capacity} require full channel state information (CSI), which is difficult to perform at BS and brings high feedback overhead in case of a large array system. To address this issue, the authors in \cite{ding2017random} implemented a random beamforming for a 2-user NOMA-MIMO cluster to avoid the requirement of full CSI at BS. The authors exploit the key feature of mmWave systems, i.e., the highly directional transmission. Thereby, they proposed a low-feedback mmWave NOMA-MIMO schemes with a fixed power allocation where the two UEs in the NOMA cluster located in the same sector are classified based on their distances from BS instead of their effective channel gains. In \cite{sun2015ergodic}, the power allocation for a Rayleigh fading 2-user NOMA-MIMO system without beamforming was investigated as an ergodic capacity maximization under statistical CSI. A m-MIMO NOMA system with limited feedback is designed in \cite{ding2016design}, where the m-MIMO NOMA channel is splited into multiple single input single output (SISO) NOMA channels, by exploiting the spatial correlation matrices of users' channels.

To reduce the BF implementation complexity and the channel overhead in m-MIMO systems, the authors in \cite{Roze2015} have been investigated the digital beamsteering (DBS). This geometric beamformer, based on the UEs' direction, steers a single beam toward each UE, using just digital phase shifters. Our previous works \cite{khaled2019performance,Israakhaled2020WSA,khaled2020jointSDMA} show that DBS is an appealing beamformer in the sparse mmWave environment trading-off complexity, performance and channel feedback. This is done using the stochastic mmWave channel model, called NYUSIM \cite{Samimi2016} built on real measurements and developed by New York University. The spatial resolution to discriminate users depends on the beamwidth, i.e. the number of antennas. However, the implementation of a large number of antennas is difficult due to high power consumption and transceiver complexity. We propose to tackle the problem of congested cell where the number of users is close to the number of antennas. To do so, we boost the DBS performance, without adding any additional antenna, by implementing PD-NOMA, as an alternative and innovative solution. Moreover, we aim to maintain a low complexity, low overhead solution, by assuming partial CSI at the BS. 

Specifically, the major contributions of this paper are summarized as follows:
\begin{itemize}
	\item We design a NOMA-DBS scheme that enhances the DBS performance without adding additional antennas, and then offers affordable hardware complexity and energy consumption. This is performed by integrating DBS with the PD-NOMA transmission, by which more than one UE can be served by each beam. The aim of the proposed NOMA-DBS scheme is to boost the DBS performance by reducing the spatial inter-beam interference (IBI). 
	\item We leverage the spatial behavior of DBS in mono- and multi-path environments, and then we define a geometric interference metric based on the UE's direction. Accordingly, we propose a user clustering algorithm for 2-UE NOMA-DBS system based on the aforementioned geometric metric instead of the effective channel gain or the channel correlation as adopted in the literature, so that reduces spatial IBI based only on the UEs' directions. 
	\item Given the set of clusters, we aim at maximizing the system throughput under transmission power and SIC constraints. Accordingly, the corresponding optimal intra-beam power allocation method is achieved by assuming a fixed inter-beam power allocation. However, the obtained scheme requires full channel feedback. 
	\item To tackle the system overhead issue, we define a partial CSI-based geometric function relying on the overall system throughput. Thereby, we reformulate a sub-optimal problem that maximizes the new geometric function, instead of the system throughput. 
	\item By simulations, we verify the performance in terms of both spectrum and energy efficiencies of the proposed NOMA-DBS scheme with either full or partial channel feedback in rural environment, using NYUSIM.
\end{itemize}

The rest of this paper is organized as follows: Section \ref{sec:systemmodel} presents the model of the proposed NOMA-DBS scheme. In Section \ref{sec:userclustering}, we investigate the spatial behavior of DBS and propose a user clustering method. Section \ref{sec:FullCSI} and Section \ref{sec:PartialCSI}, respectively, discuss the proposed full and partial CSI-based intra-cluster power allocation schemes. Section \ref{sec:simulationresults} numerically evaluates the performance of the proposed schemes, while Section \ref{sec:Conclusion} concludes the paper. \\
\textit{\textbf{Symbol Notations}}: \textbf{A}, \textbf{a} and  $ a $ denote matrix, vector and scalar, respectively.  $ \mathcal{N}(\mu,\sigma^2) $ is a Gaussian random variable with mean $ \mu $ and variance $\sigma^2$. $ (.)^T $,  $ (.)^H $ and $  \mathrm{Tr}(\cdot) $ stand for the transpose, the conjugate transpose and  the trace, respectively.

\section{Description of NOMA-DBS System}
\label{sec:systemmodel}

This study is conducted on a downlink 3 dimensions (3D) MU-MIMO system consisting of a BS equipped with a 2 dimensions (2D) array of $M= M_H\times M_V$ antennas, where $ M_H $ and $ M_V $ are the number of horizontal and vertical antennas, respectively. The BS concurrently communicates with $ K $ single-antenna UEs ($ K<M $), which are randomly distributed in the cell. Denote $ \mathcal{K}= \{1, \cdots, K\} $ as the set of UEs. Based on the angle information of the UE's direction and using digital phase shifters, classical DBS generates $ K $ beams. Each beam only serves a single UE. Thus, with space division multiple access (SDMA) strategy implemented via DBS, all UEs exploit the whole system bandwidth $ B $ for high data rate communication. 

In a congested cell, where the number of users is close to the number of antennas, the main problem of DBS is IBI, that impacts drastically the system performance in terms of sum-rate. To tackle this issue, we design a NOMA-DBS system, as illustrated in Fig. \ref{fig:SystModel}. In the proposed scheme, UEs located at the same direction are allowed to be served by one beam and the BS communicates with them by exploiting the PD-NOMA transmission protocol.

\begin{figure}[!h]
	\centerline{\includegraphics[scale=0.5]{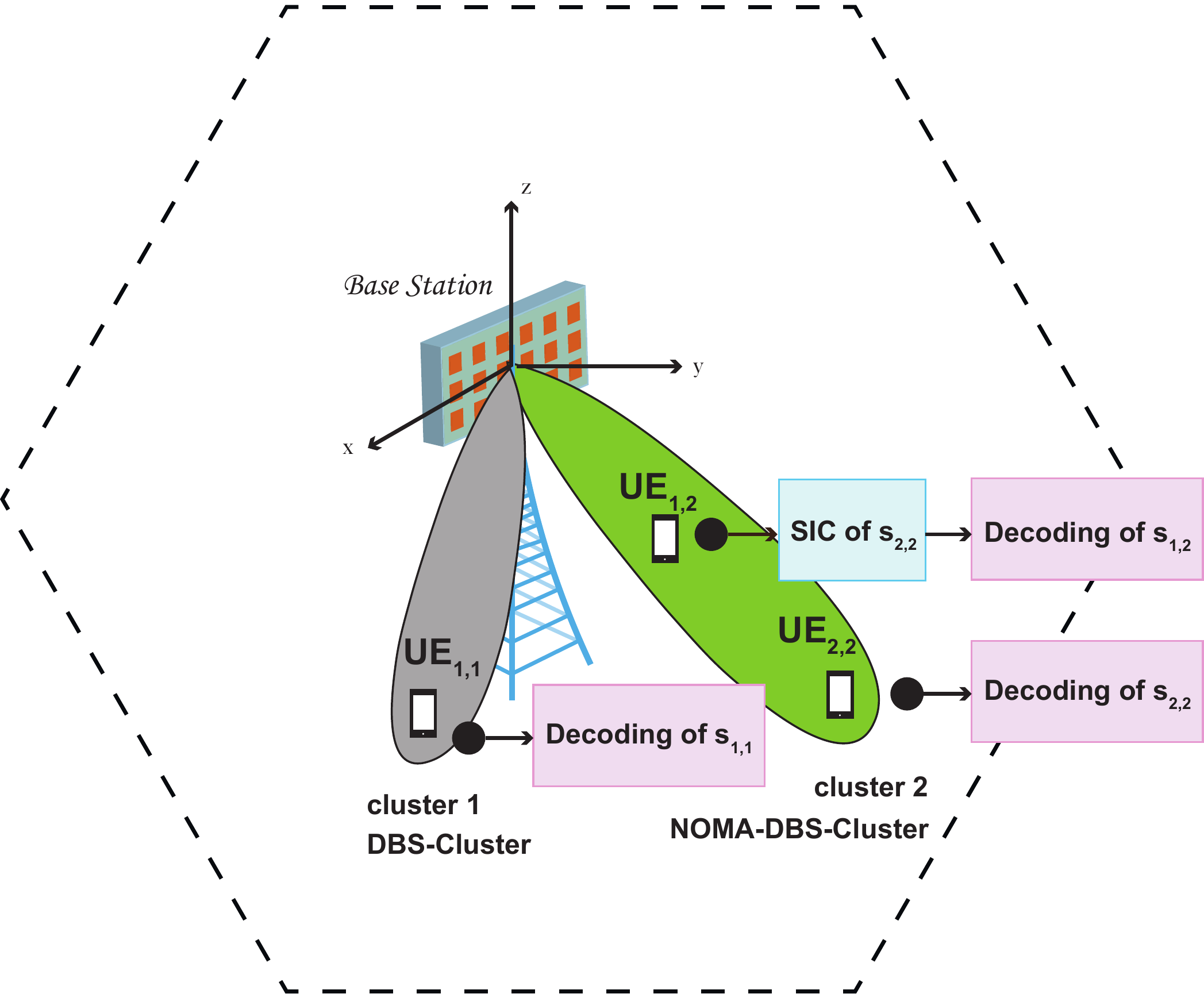}}
	\caption{NOMA-DBS beamforming system.}
	\label{fig:SystModel}
\end{figure} 

Consequently, two different types of cluster exist, namely DBS-cluster and NOMA-DBS-cluster as depicted in Fig. \ref{fig:SystModel}. In the former cluster, only one UE is served using classical DBS. However, in the latter, multiple UEs are served by the same beam and share the same time-frequency resources using PD-NOMA. Now, the $ K $ UEs are served simultaneously by $C $ $ (C\le K<M)  $ beams (or clusters).
\subsection{System model}
Let $ C_d $ and $ C_n $ be the number of DBS- and NOMA-DBS-clusters, respectively and $ K_c$ be the number of UEs in the $ c $th cluster. Hence, $ K_c=1, \ \forall c\in \mathcal{C}_d$. Denote $ \mathcal{C}_n= \{1, \cdots, C_n\} $ and  $ \mathcal{C}_d= \{C_n+1, \cdots, C\} $  as the set of DBS- and NOMA-DBS-clusters, respectively, $ \mathcal{C}=\mathcal{C}_n\cup \mathcal{C}_d $ as the set of all clusters and $ \mathcal{K}_c=\{1,\cdots, K_c\} $ as the set of UEs in the $ c $th cluster.

Based on available CSI and user clustering, detailed in section \ref{sec:userclustering}, the BS multiplexes UEs in the power domain using SC, so as to transmit the signals of UEs grouped in the same cluster within the same beam. Thus, BS constructs the superimposed signal $ \textbf{s} \in  \mathbb{C}^{C \times 1}  $, given by:
\begin{equation}
\textbf{s} =
\begin{bmatrix}
\sum_{l=1}^{K_1}\sqrt{\gamma_{l,1}p_1}s_{l,1}\\
\vdots\\
\sum_{l=1}^{K_C}\sqrt{\gamma_{l,C}p_C}s_{l,C}
\end{bmatrix}=
\begin{bmatrix}
\textbf{s}_{1}\\
\vdots\\
\textbf{s}_{C}
\end{bmatrix},
\end{equation} 
where $ s_{l,c} $ is the modulated complex symbol corresponding to the $ l $th UE in the $ c $th cluster, denoted as UE$ _{l,c} \  (l \in \mathcal{K}_c,\  c \in \mathcal{C}) $, $ p_c $ is the power allocated to the $ c $th cluster, and $ \gamma_{l,c} $ is the power allocation coefficient for $ s_{l,c} $, such that $ \sum_{l=1}^{K_c} \gamma_{l,c}=1 $. 

Moreover, according to DBS, the BS steers a beam toward each cluster in both azimuth and elevation domains. This implies that the transmit BF matrix is defined as  $ \textbf{W} =[\textbf{w}_{1}  \cdots \textbf{w}_{C}]  \in \mathbb{C}^{M \times C}$ with $\textbf{w}_{c} = \textbf{a}_{c} \in \mathbb{C}^{M \times 1}$, the BF weight vector corresponding to the $ c $th cluster. $ \textbf{a}_c=\textbf{a}(\vec{\Theta}_c) \in \mathbb{C}^{M \times 1} $ is the array steering vector corresponding to the direction $ \vec{\Theta}_{c}=(\theta_c,\phi_c)$ of the $ c $th cluster, with $ \theta $ the azimuth angle and $ \phi $ the elevation angle. In this study, we conduct a 3D MIMO channel with uniform planar array (UPA) located in $ xoz $ plane, as shown in Fig. \ref{fig:SystModel}. Therefore, the antenna array steering vector  $\textbf{a}(\vec{\Theta})=\textbf{a}(\theta,\phi)=\textbf{a}_\text{az}(\theta,\phi)\otimes\textbf{a}_\text{el}(\phi)$ is a function of  $ \theta $ and $ \phi $ with:
\begin{equation}
\label{eq:a_az}
	\textbf{a}_\text{az}(\theta,\phi)=\left[ 1, e^{j2\pi \frac{d}{\lambda}\cos(\theta)\sin(\phi)}, \cdots,  e^{j2\pi(M_H-1) \frac{d}{\lambda}\cos(\theta)\sin(\phi)} \right]^T ,
\end{equation}
\begin{equation}
\label{eq:a_el}
\textbf{a}_\text{el}(\phi)=\left[ 1, e^{j2\pi \frac{d}{\lambda}\cos(\phi)}, \cdots,  e^{j2\pi(M_V-1) \frac{d}{\lambda}\cos(\phi)} \right]^T ,
\end{equation}
where $ d $ is the antenna spacing and $ \lambda $ is the wavelength.

Then, the signal $ \textbf{x} \in  \mathbb{C}^{M \times 1}   $ transmitted from BS can be expressed as:
\begin{equation}\label{}
\textbf{x}=\sqrt{\eta}\textbf{W}\textbf{s},
\end{equation}
where $ \eta= \frac{1}{ \mathrm{Tr}({\textbf{W}}^H \textbf{W})}=\frac{1}{MC} $ is the normalization factor that eliminates the beamforming effect on the transmission power. Thus, the received signal is expressed as:
\begin{equation}\label{eq:received_signal}
{y}_{l,c} = \sqrt{\eta}\textbf{h}_{l,c}\textbf{w}_c  \textbf{s}_{c}+ \mathlarger{\sum}_{b \in \mathcal{C}, b\neq c } \sqrt{\eta}\textbf{h}_{l,c}\textbf{w}_{b}\textbf{s}_{b} + {n}_{l,c},
\end{equation}
where $ {n}_{l,c}\sim \mathcal{N}(0,\,\sigma_n^{2})\ $ is the additive white Gaussian noise, and $ \textbf{h}_{l,c} \in  \mathbb{C}^{1 \times M} $ is the multi-path 3D channel vector  between BS and UE$ _{l,c}$, which is given by: 
\begin{equation}  
\label{eq:channelvector}
\textbf{h}_{l,c}=\sum_{n=1}^{N_{l,c}} \alpha_{n,l,c} \textbf{a}^H(\vec{\Theta}_{n,l,c}),
\end{equation}
where $ N_{l,c} $ is the number of multi-path components in UE$ _{l,c}$'s channel, $ \alpha_{n,l,c} $ is the complex gain and $ \vec{\Theta}_{n,l,c} = (\theta_{n,l,c},\phi_{n,l,c}) $  is the departure angle of the $ n $th path in UE$ _{l,c}$'s channel, respectively. Denote $ n=1 $ as the index of the line-of-sight (LOS) path, i.e., the strongest path generated by NYUSIM.
\subsection{System model with SIC}

In the rest of the study, we consider only 2 UEs per NOMA-DBS cluster, such that $ K_c=2$, to limit the complexity of SIC at the receiver. The number $ K_c$ of UEs per DBS-cluster is equal to 1. 

For DBS-cluster and 2-UE NOMA-DBS-cluster, the direction $ \vec{\Theta}_{c} $ of the $ c $th cluster used to form the beams is calculated as follows:
\begin{equation}\label{eq:clusterangle}
\vec{\Theta}_c=(\theta_c,\phi_c) =
\begin{cases}
(\theta_{1,1,c},\phi_{1,1,c}) & \text{if $ K_c=1 $},\\
(\frac{\theta_{1,1,c}+\theta_{1,2,c}}{2},\frac{\phi_{1,1,c}+\phi_{1,2,c}}{2})& \text{if $ K_c=2 $}.
\end{cases}
\end{equation} 
Thus, only the departure angle of the LOS path ($n=1$) is considered.

Based on the user clustering method and the set of instantaneous received power $ |\textbf{h}_{l,c}\textbf{w}_{c}|^2 \  (l \in \mathcal{K}_c, c \in \mathcal{C}_n) $ at all UEs, the two UEs regrouped in each NOMA-DBS cluster are classified between strong and weak UEs according to the descending order of their received power, i.e., $ |\textbf{h}_{1,c}\textbf{w}_c|^2>|\textbf{h}_{2,c}\textbf{w}_c|^2 $. Moreover, SIC applied at strong UE, denoted as UE$ _{1,c}$ will work properly if the power allocation coefficients satisfy the following condition as mentioned in \cite{ali2016dynamic}:
\begin{equation}\label{eq:Const_SIC_dynamic}
\frac{\eta\gamma_{2,c}p_c|\textbf{h}_{1,c}\textbf{w}_{c}|^2}{\sigma_n^2}-\frac{\eta\gamma_{1,c}p_c|\textbf{h}_{1,c}\textbf{w}_{c}|^2}{\sigma_n^2}\ge P_{\min} \quad \forall c \in \mathcal{C}_n,
\end{equation}
where $P_{\min}$ is the minimum power difference. In fact, (\ref{eq:Const_SIC_dynamic})  maintains the power difference between strong and weak UEs' signals, so that SIC at strong UE's receiver can decode first the data of the weak UE, and then subtract it from the superimposed signal to decode his own signal. However, NOMA-DBS system also suffers from the inter-cluster interference unlike the classical PD-NOMA in \cite{ali2016dynamic}. Therefore, we propose to take into account the interference from the other clusters and (\ref{eq:Const_SIC_dynamic}) will be given by:
\begin{equation}\label{eq:Const_SIC}
\frac{\eta\gamma_{2,c}p_c|\textbf{h}_{1,c}\textbf{w}_{c}|^2}{\underbrace{\mathlarger{\sum}_{b \in \mathcal{C}, b\neq c }\eta p_b|\textbf{h}_{1,c}\textbf{w}_{b}|^2}_{I_{:\rightarrow1,c}}+\sigma_n^2}-\frac{\eta\gamma_{1,c}p_c|\textbf{h}_{1,c}\textbf{w}_{c}|^2}{\mathlarger{\sum}_{b \in \mathcal{C}, b\neq c }\eta p_b|\textbf{h}_{1,c}\textbf{w}_{b}|^2+\sigma_n^2}\ge P_{\min},
\end{equation}
where $I_{:\rightarrow l,c}$ is the interference at UE$ _{l,c}$ from all other clusters. From (\ref{eq:Const_SIC}), we find that the transmit power of UE$ _{2,c}$ must be larger than that of UE$ _{1,c}$, thus $ \gamma_{1,c}<\gamma_{2,c} $ satisfying the PD-NOMA transmission protocol. 

According to PD-NOMA principle, UE$ _{1,c}$ in NOMA-DBS cluster conducts SIC to eliminate the intra-cluster interference from UE$ _{2,c} $. We assume a perfect SIC at the receiver side of UE$ _{1,c}$. In other words, the interference at UE$_{1,c}$ from UE$ _{2,c} $, denoted as $ I_{2\rightarrow1,c}$, is totally eliminated, i.e., $ I_{2\rightarrow1,c} = 0$. Therefore, the received signal ${y}_{1,c} $ at UE$ _{1,c}$ after SIC can be expressed as:
\begin{equation}\label{eq:yc1}{y}_{1,c} = \sqrt{\eta}\sqrt{\gamma_{1,c}p_c}\textbf{h}_{1,c}\textbf{w}_c  s_{1,c}+ \mathlarger{\sum}_{b \in \mathcal{C}, b\neq c }\sqrt{\eta}\sqrt{p_b}\textbf{h}_{1,c}\textbf{w}_{b}\textbf{s}_{b} + {n}_{1,c}.
\end{equation}
And, the signal-to-interference-plus-noise ratio (SINR) $ \text{SINR}_{1,c}^{\text{NOMA-DBS}} $ of UE$ _{1,c}$ in NOMA-DBS cluster is given by:
\begin{equation}\label{eq:sinr_noma_1}
\text{SINR}_{1,c}^{\text{NOMA-DBS}} = \frac{\eta \gamma_{1,c}p_c|\textbf{h}_{1,c}\textbf{w}_c|^2}{   \mathlarger{\sum}_{b \in \mathcal{C}, b\neq c }\eta p_b|\textbf{h}_{1,c}\textbf{w}_{b}|^2+  \sigma^2_n}.
\end{equation}

However, UE$ _{2,c} $ of NOMA-DBS cluster treats the UE$ _{1,c}$ signal as an interference and the received signal $ {y}_{2,c} $ at UE$ _{2,c} $ can be expressed as:
\begin{multline}
{y}_{2,c} = 
\underbrace{\sqrt{\eta}\sqrt{\gamma_{2,c}p_c}\textbf{h}_{2,c} \textbf{w}_cs_{2,c}}_{\text{Desired signal}}+ \underbrace{\sqrt{\eta}\sqrt{\gamma_{1,c}p_c}\textbf{h}_{2,c} \textbf{w}_cs_{1,c}}_{\text{Intra-beam interference}}+ \underbrace{\mathlarger{\sum}_{b \in \mathcal{C}, b\neq c }\sqrt{\eta} \sqrt{p_b} \textbf{h}_{2,c}\textbf{w}_{b}\textbf{s}_{b}}_{\text{Inter-beam interference}}+{n}_{2,c}.
\end{multline} 
Therefore, the SINR $ \text{SINR}_{2,c}^{\text{NOMA-DBS}} $ of UE$ _{2,c}$ is given by:
\begin{equation}\label{eq:sinr_noma_2}
\text{SINR}_{2,c}^{\text{NOMA-DBS}} = \frac{\eta \gamma_{2,c}p_c|\textbf{h}_{2,c}\textbf{w}_c|^2}{ \eta \gamma_{1,c}p_c|\textbf{h}_{2,c}\textbf{w}_c|^2 +  \mathlarger{\sum}_{b \in \mathcal{C}, b\neq c}\eta p_b |\textbf{h}_{2,c}\textbf{w}_{b}|^2 +  {\sigma^2_n}}.
\end{equation} 

For DBS-cluster, the received signal $ \textbf{y}_{1,c} $ and $ \text{SINR}_{1,c}^{\text{DBS}} $ are also calculated as in (\ref{eq:yc1}) and (\ref{eq:sinr_noma_1}), respectively, with $ \gamma_{1,c}=1 $.

The SINR achieved at UE$_{l,c}$ in DBS and NOMA-DBS clusters can be rewritten as:
\begin{equation}\label{eq:sinr_1c_dbs_fct_zeta}
\text{SINR}_{1,c}^{\text{DBS}}=\frac{\psi_{1,c}}{\nu_{1,c}}= \zeta_{1,c},
\end{equation}
\begin{equation}\label{eq:sinr_1c_fct_zeta}
\text{SINR}_{1,c}^{\text{NOMA-DBS}}=\frac{\psi_{1,c}\gamma_{1,c}}{\nu_{1,c}}= \zeta_{1,c}\gamma_{1,c},
\end{equation}
\begin{equation}\label{eq:sinr_2c_fct_zeta}
\text{SINR}_{2,c}^{\text{NOMA-DBS}}=\frac{\psi_{2,c}(1-\gamma_{1,c})}{\nu_{2,c}+\psi_{2,c}\gamma_{1,c}}= \frac{\zeta_{2,c}(1-\gamma_{1,c})}{1+\zeta_{2,c}\gamma_{1,c}},
\end{equation}
where 
\begin{equation}\label{}
\psi_{l,c}=\eta p_c|\textbf{h}^H_{l,c}\textbf{w}_c|^2, \quad l \in \mathcal{K}_c, c \in \mathcal{C}_n,
\end{equation}
\begin{equation}\label{}
\nu_{l,c}=\mathlarger{\sum}_{b \in \mathcal{C}, b\neq c }\eta p_b|\textbf{h}^H_{l,c}\textbf{w}_{b}|^2 +  \sigma^2_n, \quad l \in \mathcal{K}_c, c \in \mathcal{C}_n,
\end{equation}
 \begin{equation}\label{}
\zeta_{l,c}=\frac{\psi_{l,c}}{\nu_{l,c}}, \quad l \in \mathcal{K}_c, c \in \mathcal{C}_n.
\end{equation}
As seen in (\ref{eq:received_signal}), $\psi_{l,c}$ represents the received power at UE$_{l,c}$ of the superimposed signal, that is transmitted in the $c$th cluster. And, $ \nu_{l,c} $ represents the interference $I_{:\rightarrow l,c}$ at UE$_{l,c}$ coming from the other clusters plus the noise power. This implies that $ \zeta_{l,c} $ is the superimposed-signal to other-clusters interference plus noise ratio at  UE$_{l,c}$, which is different than SINR$_{l,c}$ for NOMA-DBS cluster and represents SINR$_{1,c}$ for DBS-cluster. Therefore, (\ref{eq:Const_SIC}) can be simplified as follows:
\begin{equation}\label{eq:SIC_constraint}
\gamma_{2,c}-\gamma_{1,c}\ge \frac{P_{\min}}{\zeta_{1,c}}.
\end{equation}

\section{3D Geometric Interference Metric $\beta$ and User Clustering Method ($\beta$-UC)}
\label{sec:userclustering}
In this section, we discuss the factors affecting the interference level at each UE in case of classical DBS, in both mono- and multi-path environments, to define later a 3D geometric interference metric. Subsequently, we propose a user clustering strategy based on the aforementioned metric. 
\subsection{Classical DBS Performance}
\label{ssec:DBSperf}
In case of classical DBS, we use the subscript $ k $ $ (\forall k \in \mathcal{K}) $ instead of $ (l,c )$.
\subsubsection{Mono-path environment}~ \\
In mono-path environment, the SINR $ \text{SINR}_k^{\text{DBS}} $ achieved using DBS at UE$ _k $ $ (\forall k \in \mathcal{K}) $ is given by \cite{Israakhaled2020WSA}:
\begin{equation}\label{eq:sinr_dbs_mono}
{\text{SINR}_k^{\text{DBS}}=\frac{|\textbf{a}^H_{1,k}\textbf{a}_{1,k}|^2}{\underbrace{\mathlarger{\sum}_{u\in \mathcal{K}, u\neq k} |\textbf{a}^H_{1,k}\textbf{a}_{1,u}|^2}_{I_k} + \frac{\sigma^2_n}{\eta^{\text{DBS}}|\alpha_{1,k}|^2}}},
\end{equation}
where $\textbf{a}_{1,k}=\textbf{a}(\vec{\Theta}_{1,k})$ is the steering vector corresponding to the LOS path angle of UE$ _k $, denoted as $ \vec{\Theta}_{1,k}=(\theta_{1,k},\phi_{1,k}) $, $ \alpha_{1,k} $ is the complex gain of the LOS path in UE$ _{k}$'s channel, and $ \eta^{\text{DBS}} $ is the normalization factor when BS adopts only DBS.

For classical DBS in mono-path environment, the interference term $I_k$ at UE$ _k $, observed at the
denominator in (\ref{eq:sinr_dbs_mono}), only depends on $\mu_{k,u}\stackrel{\text{def}}{=}|\textbf{a}^H_{1,k}\textbf{a}_{1,u}|$ $(\forall u\neq k)$, i.e., the spatial interference causing by the beam generated for UE$ _u $ with the LOS path in the UE$_k$'s channel. In other words, if UE$ _k $ and UE$ _u $ are located at the same direction, i.e., $ \vec{\Theta}_{1,k}$ is closed to $ \vec{\Theta}_{1,u}$, then DBS system suffers from high IBI, and thus the DBS performance degrades. 
\subsubsection{Multi-path environment} ~ \\ 
Similarly, in multi-path environment, DBS also generates a single beam
in the LOS path direction. Therefore, $\text{SINR}_k^{\text{DBS}}$ is given by \cite{Israakhaled2020WSA}:
\begin{equation}
\label{eq:sinr_dbs_multi}
\text{SINR}_k^{\text{DBS}}=\frac{\left| \textbf{a}^H_{1,k}\textbf{a}_{1,k} +{\mathlarger{\mathlarger{\sum}}_{n=2}^{N_k}  \frac{\alpha_{n,k}}{\alpha_{1,k}}\textbf{a}^H_{n,k}\textbf{a}_{1,k}}\right|  ^2}{ \underbrace{ \mathop{\mathlarger{\sum}_{u\in \mathcal{K}, u\neq k}}\left|  \mathlarger{\mathlarger{\sum}}_{n=1}^{N_k}\frac{\alpha_{n,k}}{\alpha_{1,k}}{\textbf{a}^H_{n,k}\textbf{a}_{1,u}}\right| ^2}_{I_k} + \frac{\sigma^2_n}{\eta^{\text{DBS}}|\alpha_{1,k}|^2}},
\end{equation}
where $ \alpha_{n,k} $,  $ \vec{\Theta}_{n,k} $ and $ \textbf{a}_{n,k}=\textbf{a}(\vec{\Theta}_{n,k}) $ are complex gain, angle and  steering vector of the $ n $th path in UE$ _k$'s channel, respectively.

As seen in (\ref{eq:sinr_dbs_multi}), the interference $I_k$ at UE$ _k $ is caused by the LOS beam of UE$ _u $  ($\forall u\neq k$), represented by $\textbf{a}_{1,u}$, that interferes with all paths of  UE$ _k $, represented by $\textbf{a}^H_{n,k}$. Thus, the interference results from both non-LOS (NLOS) and LOS paths. However, DBS is a geometric beamformer, that doesn't exploit the NLOS paths and generates a single beam in the LOS path direction, i.e., the directions of the UEs. %This is the main advantage of using DBS, that enables a BF transmission with partial channel feedback, mainly the UEs' directions. For that, in multi-path environment, we don't consider the total interference $ I_k $, and we make the simplifying assumption that investigates the interference with the LOS path without taking into account the NLOS paths. 
This simplifying assumption makes possible partial CSI and feedback solely based on the UEs' directions but restricts the sum rate, since all the path diversity is not exploited. 

\subsection{ Interference Metric}
\label{ssec:perfMetric}
To enhance the DBS performance, our proposed  NOMA-DBS scheme divides UEs into clusters based on their spatial direction, which directly reflects in their spatial interferences.

Based on the simplifying assumption and the previous analysis in Section \ref{ssec:DBSperf}, we define $ \beta_{k,u}=\frac{1}{M}\mu_{k,u}= \frac{1}{M}|\textbf{a}^H_{1,k}\textbf{a}_{1,u}| = \frac{1}{M}|\textbf{a}^H_{1,u}\textbf{a}_{1,k}|=\frac{1}{M}\mu_{u,k}=\beta_{u,k}$ as a 3D geometric metric representing the normalized spatial interference at UE$_k$ from the beam generated for UE$_u$, and vice versa. 
Applying (\ref{eq:a_az}) and (\ref{eq:a_el}), using UPA,  $ \beta_{k,u} $ is given by:
	\begin{subequations}
	\begin{align}
	\beta_{k,u} & = \frac{1}{M}\left| \mathlarger{\sum}_{m_H=1}^{M_H}\sum_{m_V=1}^{M_V}e^{j2\pi\frac{d}{\lambda}\{(m_H-1)(\cos(\theta_{1,k})\cos(\phi_{1,k})-\cos(\theta_{1,u})\cos(\phi_{1,u}))+(m_V-1)(\sin(\phi_{1,k})-\sin(\phi_{1,u}))\}}\right| 
	\end{align}
	\begin{align}
	&  = \left| \frac{\sin\left( \frac{M_H\pi d}{\lambda}(\cos(\theta_k)\cos(\phi_k)-\cos(\theta_u)\cos(\phi_u))\right) }{M_H\sin\left( \frac{\pi d}{\lambda}(\cos(\theta_k)\cos(\phi_{1,k})-\cos(\theta_{1,u})\cos(\phi_{1,u}))\right) }\frac{\sin\left( \frac{M_V\pi d}{\lambda}(\sin(\phi_{1,k})-\sin(\phi_{1,u})\right) }{M_V\sin\left( \frac{\pi d}{\lambda}(\sin(\phi_{1,k})-\sin(\phi_{1,u})\right) }\right|  
	\end{align}.
	\label{eq:beta}
\end{subequations}
According to (\ref{eq:beta}), $ \beta_{k,u} \in (0,1)$, and $\beta_{k,u}$ represents the value of the normalized array pattern  $ {AF}^n_{(\theta_{1,u},\phi_{1,u})}(\theta,\phi) $ \cite{hansen2009phased} of the beam pointed at UE$ _u $ direction $ (\theta_{1,u},\phi_{1,u}) $ for $ (\theta,\phi)=(\theta_{1,k},\phi_{1,k}) $, i.e.,
\begin{equation}\label{eq:beta_AF}
\beta_{k,u}={AF}^n_{(\theta_{1,u},\phi_{1,u})}(\theta_{1,k},\phi_{1,k}).
\end{equation} 

\begin{remark}
Based on (\ref{eq:beta_AF}), we can see that if  the angular distance $ (\Delta\theta,\Delta\phi)\stackrel{\text{def}}{=} (|\theta_{1,k}- \theta_{1,u}|,|\phi_{1,k}- \phi_{1,u}|) \rightarrow (0,0)$ , i.e., $ (\theta_{1,u},\phi_{1,u}) \rightarrow  (\theta_{1,k},\phi_{1,k})$, then $  \beta_{k,u}  \rightarrow 1 $. Thus, the smaller $ \Delta\theta $ and $ \Delta\phi $ are, i.e., UE$ _k $ and UE$ _u $ are closely at the same direction in azimuth and elevation domains, the larger $ \beta_{k,u} $ is. 
\end{remark}

\begin{remark}
$ {AF}^n_{(\theta_{1,u},\phi_{1,u})}(\theta,\phi) = \sqrt{\frac{1}{2}}$ gives a measure of the 3dB-beamwidth $ {\vec{\Omega}_u}^{3\text{dB}}=({{\Omega_\text{az}}}_u^{3\text{dB}},{{\Omega_\text{el}}}_u^{3\text{dB}})  $ expressed in terms of the azimuth $ {{\Omega_\text{az}}}_{u} $ and elevation $ {\Omega_\text{el}}_{u} $ widths of the beam generated at UE$ _u $, which depend on the beam direction $\vec{\Theta}_u$ and the number of horizontal and vertical antennas. $ {{\Omega_\text{az}}}_u^{3\text{dB}} = |\theta_{1,u} - \theta| $ and $ {{\Omega_\text{el}}}_u^{3\text{dB}} = |\phi_{1,u} - \phi| $ are defined as the angular distances that satisfy $ {AF}^n_{(\theta_{1,u},\phi_{1,u})}(\theta,\phi) =\sqrt{\frac{1}{2}}$. This implies that $ \beta_{k,u}$ implicitly provides information regarding the beamwidth $ {\vec{\Omega}_u}^{3\text{dB}} $. 
\end{remark}

Therefore, we define the interference threshold $ \beta_0 $, such that $ \beta_{k,u}\ge \beta_0 $ means that the LOS path of UE$ _k $ lies in the UE$ _u $ beam in both azimuth and elevation domains. Particularly, the $ \beta_0 $-beamwidth $\vec{\Omega}_u^{\beta_0}$ defines the azimuth and elevation angular distances satisfying $ |{AF}^n_{(\theta_{1,u},\phi_{1,u})}(\theta,\phi)|=\beta_0 $. Thus, $ \beta $ is an important metric to determine the level of spatial interference based on a geometric partial channel knowledge, namely the UE direction, i.e., $ (\theta_{1,k},\phi_{1,k}) $ ($\forall k \in \mathcal{K}$).
\subsection{User clustering algorithm}

We would like to apply NOMA to UEs that produce high spatial interference to each other, whereas the others are treated via DBS. To perform user clustering, we take profit of the 3D geometric interference based on the angles of the LOS path of each user. To the best of our  knowledge, it is the first time that such a geometric metric is used in user clustering. The main idea of our proposed UC, denoted as $ \beta$-UC, is to regroup two-by-two UEs (UE$ _u $ and UE$ _k $) with large inter-beam interference, i.e., $ \beta_{k,u} >\beta_0 $, into the same cluster, whose beam angle is defined in \ref{eq:clusterangle}. Subsequently, we assign a single cluster to the remaining UEs with cluster angle equal to their LOS path angle, as mentioned in (\ref{eq:clusterangle}). The proposed $ \beta$-UC algorithm can be given as in Algorithm \ref{alg:UserClsuteringAlgorithm}.
\LinesNumberedHidden
\begin{algorithm}[]
	\DontPrintSemicolon	
	\KwRequire{ $ \vec{\Theta}_{1,k} = (\theta_{1,k},\phi_{1,k}) $, $ \beta_{k,u} \  \forall  k,u< k \in \mathcal{K}$, $ \beta_0 $. }
	\KwOutput{ $ \mathcal{G} = \left\lbrace \textbf{g}_c\right\rbrace$: set of clusters with $ \textbf{g}_c $ the array of UEs belonging to the $ c $th cluster.}
	\KwStep{Select the groups of 2-UEs who have a spatial interference greater than $ \beta_0 $;  \newline $ \mathcal{D}=\{(k,u), \ \beta_{k,u}\ge \beta_0, k<u \in \mathcal{K}\}$.}
	\KwStepp{Select from $ \mathcal{D} $ the subset of 2-UEs having the largest spatial interference;  \newline $ (l,m)=\underset{(k,u)\in \mathcal{D}}{\max }\{\beta_{k,u}\}, $  $\mathcal{G}=\mathcal{G}\cup [l,m] $.}
	\KwSteppp{Remove from $ \mathcal{D} $ the spatial interference of UEs selected in Step 2 with each other and with any other UEs to prevent their existence in another cluster;  \newline $\mathcal{D}\leftarrow \mathcal{D}- (l,m) - (l,w)-(v,l) - (m,q) - (t,m) $ $ (\forall w>l,v<l,q>m,t<m \in \mathcal{K})$.}
	\KwStepppp{Repeat Step 2 and Step 3 until $ \mathcal{D}=\emptyset $.}
	\KwSteppppp{Assign for each one of the remaining UEs a beam in the corresponding UE's direction; \newline $  s \in \mathcal{K},  s \notin \mathcal{G}, \mathcal{G}=\mathcal{G}\cup [s] $. }
	\caption{$ \beta$-based User Clustering Algorithm}	
	\label{alg:UserClsuteringAlgorithm}
\end{algorithm}

\section{Full CSI-based Power Allocation Scheme}
\label{sec:FullCSI}

In this section, we design a power allocation (PA) scheme based on full CSI, that maximizes the throughput of the NOMA-DBS system. However, DBS is a geometric beamformer that requires only the UE's direction. For that, it's more significant to design a PA scheme based on partial CSI, which is presented in Section \ref{sec:PartialCSI}.
\subsection{Problem Formulation}
The system throughput,  which is defined as the maximum quantity of the transmitted information over the NOMA-DBS system, is given by a sum of the UEs' data rates:
\begin{subequations}
	\begin{align}
	R_T=\sum_{k\in \mathcal{K}}R_k=\sum_{c\in \mathcal{C}}\sum_{l\in\mathcal{K}_c}R_{l,c} 
	\end{align}
	\vspace*{-3mm}
	\begin{align}
	\label{}
	&  =\underbrace{\sum_{c\in \mathcal{C}_d}R_{1,c}}_{\text{DBS clusters}}+\underbrace{\sum_{c\in \mathcal{C}_n}(R_{1,c}+R_{2,c})}_{\text{NOMA-DBS clusters}},
	\end{align}
\end{subequations}
where $ R_{l,c} $ [$ bps $] is the achievable data rate for UE$ _{l,c}$ and can be expressed, according to Shannon formula, as:
\begin{equation}\label{}
R_{l,c}= B\log_2(1+\text{SINR}_{l,c}), \quad   l \in \mathcal{K}_c, c \in \mathcal{C}_d\cup\mathcal{C}_n.
\end{equation} 
Therefore, $ R_T $ can be defined as a function of the inter-cluster PA $\boldsymbol{p}= \{p_c, \ c \in \mathcal{C}\} $ and the intra-cluster PA $\boldsymbol{\gamma}= \{\gamma_{l,c},\ l\in \{1,2\}, \ c \in \mathcal{C}_n\} $, since $ \gamma_{1,c} =1$  $ (\forall c \in \mathcal{C}_d)$, i.e., $ R_T = F(\boldsymbol{\gamma},\boldsymbol{p}) $:
\begin{equation}\label{}
R_T  = F^{\text{DBS}}(\boldsymbol{\gamma},\boldsymbol{p}) + F_{1}^{\text{NOMA-DBS}}(\boldsymbol{\gamma},\boldsymbol{p}) + F_{2}^{\text{NOMA-DBS}}(\boldsymbol{\gamma},\boldsymbol{p}),
\end{equation}  
where $ F^{\text{DBS}} $ is the throughput of UEs grouped in DBS-clusters, $ F_{1}^{\text{NOMA-DBS}} $ is the throughput of strong UEs grouped in NOMA-DBS-clusters and $ F_{2}^{\text{NOMA-DBS}} $ is the throughput of weak UEs grouped in NOMA-DBS-clusters. Using (\ref{eq:sinr_1c_dbs_fct_zeta}), (\ref{eq:sinr_1c_fct_zeta}) and (\ref{eq:sinr_2c_fct_zeta}), $ F^{\text{DBS}} $, $ F_{1}^{\text{NOMA-DBS}} $ and $ F_{2}^{\text{NOMA-DBS}} $ are given by (\ref{eq:sumRate_DBS_1c}), (\ref{eq:sumRate_NOMADBS_1c}) and (\ref{eq:sumRate_NOMADBS_2c}), respectively:
\begin{equation}
\label{eq:sumRate_DBS_1c}
F_{1}^{\text{NOMA-DBS}}(\boldsymbol{\gamma},\boldsymbol{p})=B\sum_{c\in \mathcal{C}_n}\log_2(1+\zeta_{1,c}(\boldsymbol{p}) ),
\end{equation}
\begin{equation}
\label{eq:sumRate_NOMADBS_1c}
F_{1}^{\text{NOMA-DBS}}(\boldsymbol{\gamma},\boldsymbol{p})=B\sum_{c\in \mathcal{C}_n}\log_2(1+ \zeta_{1,c}(\boldsymbol{p})\gamma_{1,c}),
\end{equation}
\begin{equation}
\label{eq:sumRate_NOMADBS_2c}
F_{1}^{\text{NOMA-DBS}}(\boldsymbol{\gamma},\boldsymbol{p})=B\sum_{c\in \mathcal{C}_n}\log_2\left( 1+\frac{\zeta_{2,c}(\boldsymbol{p})(1-\gamma_{1,c})}{1+\zeta_{2,c}(\boldsymbol{p})\gamma_{1,c}}\right). 
\end{equation}

The aim of the proposed PA scheme is to maximize the sum of the UEs' data rates over the NOMA-DBS system. Therefore, the corresponding optimization problem denoted by $ P1 $ is formulated as follows:
\begin{equation}\label{pb:OptimizedPA} 
\begin{aligned}P1: \quad \{\stackrel{\star}{{\boldsymbol{\gamma}}},\stackrel{\star}{{\boldsymbol{p}}}\}= &
\underset{\boldsymbol{\gamma},\boldsymbol{p}}{\max}  \ F(\boldsymbol{\gamma},\boldsymbol{p}),\\
\text{s.t.}:      & \ C_1: {\gamma_{1,c}}+{\gamma_{2,c}}=1 \ \forall c \in \mathcal{C}_n,\\
& \ C_2: \gamma_{1,c} \le \gamma_{2,c} \ \forall c \in \mathcal{C}_n,\\
& \ C_3:  \gamma_{2,c}-\gamma_{1,c}\ge \frac{P_{\min}}{\zeta_{1,c}} \ \forall c \in \mathcal{C}_n.\\
%& \ C_4: \sum_{c=1}^{C}P_c=P_e \ \forall c \in \mathcal{C}_n \\
\end{aligned}
\end{equation}
$ P1 $ must be done under the power allocation constraints $ C_1 $ and $ C_2 $ and the SIC constraint $ C_3 $ obtained in (\ref{eq:SIC_constraint}).

Problem $ P1 $ maximizes the system throughput by jointly optimizing intra- and inter-cluster PA. However, because the coupling of the power allocation factors from different UEs, then it is difficult to obtain the optimal solutions. To this end, we add a constraint that addresses this issue and then simplifies $ P1 $ by applying a fixed inter-cluster PA. Thus once $\beta$-UC is applied, the power allocated to each cluster $ \boldsymbol{p} $ will be well known and equal to constant values as detailed in Section \ref{ssec:fixedInterClusterPA}. Hence, $ P1 $ can be reformulated as follows:
\begin{equation}\label{pb:OptimizedPA_simplified} 
\begin{aligned}P1: \quad \stackrel{\star}{{\boldsymbol{\gamma}}}= &
\underset{\boldsymbol{\gamma}}{\max}  \ F(\boldsymbol{\gamma}),\\
\text{s.t.}:      & \ C_1: {\gamma_{1,c}}+{\gamma_{2,c}}=1 \ \forall c \in \mathcal{C}_n,\\
& \ C_2: \gamma_{1,c} \le \gamma_{2,c} \ \forall c \in \mathcal{C}_n,\\
& \ C_3:  \gamma_{2,c}-\gamma_{1,c}\ge \frac{P_{\min}}{\zeta_{1,c}} \ \forall c \in \mathcal{C}_n.
\end{aligned}
\end{equation}

For DBS cluster, as seen in (\ref{eq:sumRate_DBS_1c}), the throughput of UEs grouped in DBS-cluster is the same regardless of the intra-cluster PA method, i.e., $  F^{\text{DBS}}(\boldsymbol{\gamma}) = \text{constant} $ ($ \forall \boldsymbol{\gamma} $). However, for NOMA-DBS cluster, as seen in (\ref{eq:sumRate_NOMADBS_1c}) and (\ref{eq:sumRate_NOMADBS_2c}), the data rates of UE$ _{1,c} $ and UE$ _{2,c} $ only depend on their PA coefficients $ \gamma_{1,c} $ and $ \gamma_{2,c}=1-\gamma_{1,c}$. To this end, based on $ \beta$-UC, the PA scheme maximizing the system throughput is equivalent to that maximizing the throughput of each NOMA-DBS cluster separately. Accordingly, the maximization problem  $ P1 $ can be seen as an equivalent problem of $ C_n $ independent optimization problems denoted by $ P2^c $ for the $ c^{th} $ NOMA-DBS cluster, and can be equivalently formulated as follows:
\begin{equation}\label{pb:OPA_simplified} 
\begin{aligned} 
P2^c : \quad \stackrel{\star}{\gamma}_{1,c}= & \underset{\gamma_{1,c} }{\max} \ R_c= \underset{\gamma_{1,c} }{\max} \ R_{1,c}+R_{2,c} \ \forall c \in \mathcal{C}_n, \\
\text{s.t}:      & \ C_1: {\gamma_{1,c}}+{\gamma_{2,c}}=1 \ \forall c \in \mathcal{C}_n, \\
&\  C_2': \gamma_{1,c} \le \frac{1}{2},  \\
&\ C_3: \gamma_{2,c}-\gamma_{1,c}\ge  \frac{P_{\min}}{\zeta_{1,c}}.\\
\end{aligned}
\end{equation}
$C_1$ and $C_2$ in problem (\ref{pb:OptimizedPA}) can be reformulated as $C_2'$. Moreover, $C_1$, $C_2'$ and $C_3$ specify the feasible interval $ \Gamma_c $ of $ \gamma_{1,c} $, that is $ \Gamma_c = [0,\hat{\gamma}_c]$ with:
\begin{equation}
\hat{\gamma}_c = \frac{1}{2}\left( 1- \frac{P_{\min}}{\zeta_{1,c}}\right). 
\end{equation}
\subsection{Fixed Inter-Cluster Power Allocation}
\label{ssec:fixedInterClusterPA}
In our previous work \cite{khaled2020jointSDMA}, we applied a uniform power allocation per cluster, i.e., $ p_c = \text{constant} \quad (\forall c)$. We define $ P_c $ the emitted power toward the $ c $th cluster as follows:
\begin{equation}\label{eq:Pc_emittedPower}
P_c = \eta \lVert \mathbf{w}_c \rVert^2\sum_{l=1}^{K_c}\gamma_{l,1}p_c=\eta \lVert \mathbf{w}_c \rVert^2p_c.
\end{equation}
Moreover, $ \sum_{c=1}^{C}P_c=P_e $ where $ P_e $ is the total transmit power. This implies that the uniform power allocation per cluster satisfies $P_c=\frac{P_e}{C}$, using $ \eta \lVert \mathbf{w}_c \rVert^2=\frac{1}{C} $. However, in this study, we allocate a power $ p_c $ to the $ c $th cluster proportional to the number $ K_c $ of UEs served in the $ c $th cluster. Accordingly, based on (\ref{eq:Pc_emittedPower}), $ P_c $ is also proportional to $ K_c $. This is supposed to guarantee the power fairness among UEs. Therefore, we calculate $ P_c $ as follows: 
\begin{equation}\label{eq:Pc_noUniformPA}
P_c =K_c\frac{P_e}{K}.
\end{equation}
From (\ref{eq:Pc_emittedPower}) and (\ref{eq:Pc_noUniformPA}) we obtain that:
\begin{equation}\label{}
p_c=K_cC\frac{P_e}{K}.
\end{equation}
\subsection{Intra-Cluster Power Allocation Solution with Full CSI}
\label{ssec:OptimalSolutionFullCSI}
In this subsection, we propose an intra-cluster PA scheme, for the problem $ P2^c $ for each NOMA-DBS cluster, denoted as OPA, which is summarized in Algorithm \ref{alg:Algorithm1}.
\begin{algorithm}[!h]	
	\DontPrintSemicolon	
	\textbf{initialization:}
	$ c=C_n $;\\
	\While{$ c\le C $}
	{   
		\If{$ \left| \frac{\log_2(1+\zeta_{1,c})-\log_2(1+\zeta_{2,c})}{\log_2(1+\zeta_{1,c})}\right| < \epsilon $}{$\gamma_{1,c}= \tilde{\gamma}_{1,c}$}
		%		\tcp*{\scriptsize  $ R_c $ is constant }
		\ElseIf{$\zeta_{2,c}==\min(\zeta_{1,c},\zeta_{2,c}) $}{$\gamma_{1,c}=\hat{\gamma}_c$}
		%		\tcp*{\scriptsize  $ R_c $ is an increasing function}
		\Else{$\gamma_{1,c}=0$}
		c=c+1;			
	}		
	\caption{Optimal Power Allocation (OPA)}
	\label{alg:Algorithm1}
\end{algorithm}
\begin{lemma}
	\label{lem:RcMonotone}
	$R_c$ is a real monotone function of $\gamma_{1,c}$, and  its variation depends on the sign of $\zeta_{1,c}-\zeta_{2,c}$.
\end{lemma}
\begin{proof}
	See Appendix 1.
\end{proof}

From Lemma \ref{lem:RcMonotone}, $ R_c $ is a real monotone function of $\gamma_{1,c}$, and its variation is as follows: 
\begin{equation}\label{eq:Rc_sign}
\left\lbrace 
\begin{aligned}
\zeta_{1,c}> \zeta_{2,c} \Longrightarrow R_{c} \ \text{is an increasing function of  } \ \gamma_{1,c}\\
\zeta_{1,c}< \zeta_{2,c} \Longrightarrow R_{c} \ \text{is an decreasing function of  } \ \gamma_{1,c}\\
\zeta_{1,c}= \zeta_{2,c} \Longrightarrow R_{c} \ \text{is constant  } \ \forall \gamma_{1,c}
\end{aligned}
\right. 
\end{equation}

For the three cases, and since UE$_{1,c}$ is the strong UE, i.e., $|\textbf{h}_{1,c}\textbf{w}_c|^2>|\textbf{h}_{2,c}\textbf{w}_c|^2$, this implies that $\psi_{1,c}>\psi_{2,c}$. If $\zeta_{1,c}<\zeta_{2,c}$, then $I_{:\rightarrow 1,c}>I_{:\rightarrow 2,c}$. And, UE$_{1,c}$ taken as a strong UE suffers from high other-clusters interference compared to UE$_{2,c}$. This implies that even though UE$_{1,c}$ has the highest received superimposed signal, the level of interference $ I_{:\rightarrow l,c} $ at each UE from other clusters determines the variation of the function $R_c$.

From Lemma \ref{lem:RcMonotone} and based on the derivation and analysis in Section \ref{ssec:OptimalSolutionFullCSI}, the optimal solution $ \stackrel{\star}{\gamma}_{1,c} $ of Problem $ P2^c $ is achieved at the end point of the feasible interval $ \Gamma_c $ of $ \gamma_{1,c} $. Indeed, if $ \zeta_{1,c}>\zeta_{2,c} $, i.e., $ R_c $ is an increasing function, then $ \stackrel{\star}{\gamma}_{1,c} = \hat{\gamma}_c $. The authors in \cite{ali2016dynamic} propose an optimal PA policy that maximizes the sum-throughput of a classical NOMA system (NOMA-SISO) by maximizing the sum-throughput per m-user NOMA cluster ($2 \le m \le K $). Using Karush-Kuhn-Tucker optimality conditions, they derived the optimal PA. We find that for 2-UE NOMA-SISO cluster (see Table 1 \cite{ali2016dynamic}), the optimal PA can be similar to our solution obtained without using Lagrange, by  changing variables. Otherwise, if $ \zeta_{1,c} < \zeta_{2,c} $, i.e., $ R_c $ is a decreasing function, then $ \stackrel{\star}{\gamma}_{1,c} = 0 $. This implies that OPA deactivates UE$ _{1,c} $, which is classified as strong UE, but actually suffers from high inter-cluster interference, by allocating all the power $ p_c $ to UE$ _{2,c} $. In addition, when $ |\zeta_{1,c}-\zeta_{2,c}|<\epsilon $, each value belonging to the interval $ \Gamma_c $ of $\gamma_{1,c}$ can be taken as an optimal solution. But, we choose $\tilde{\gamma}_{1,c}$ that achieves fairness between UEs in the $ c $th cluster, i.e., $ R_{2,c}(\tilde{\gamma}_{1,c})=R_{1,c}(\tilde{\gamma}_{1,c}) $, thus $ \stackrel{\star}{\gamma}_{1,c} = \tilde{\gamma}_{1,c} $ with 
\begin{equation}\label{eq:gamma_tiled}
\tilde{\gamma}_{1,c}=\frac{-(\zeta_{1,c}+\zeta_{2,c})+\sqrt{(\zeta_{1,c}+\zeta_{2,c})^2+4\zeta_{1,c}\zeta_{2,c}^2}}{2\zeta_{1,c}\zeta_{2,c}}.
\end{equation}
\begin{lemma}
	\label{lem:boundGamma1}
	When $ R_c $ is a constant function, then $ \tilde{\gamma}_{1,c} $ chosen as the optimal solution of Problem $ P2 $ satisfies the power allocation constraint $ C_1' $. 
\end{lemma}
\begin{proof}
	As seen in (\ref{eq:Rc_sign}), $ R_c $ is a constant function if $\zeta_{1,c} = \zeta_{2,c} $. Therefore, for $ \zeta = \zeta_{1,c} \rightarrow  \zeta_{2,c} $, $ \tilde{\gamma}_{1,c} $ can be expressed as:
	\begin{equation}\label{eq:boundGamma1}
	\tilde{\gamma}_{1,c} = \frac{-1+\sqrt{1+\zeta}}{\zeta}.
	\end{equation}
	Applying the first-order derivative of (\ref{eq:boundGamma1}), we obtain that:
	\begin{equation}\label{eq:tiled_gamma_diff}
	\diff{\tilde{\gamma}_{1,c}}{\zeta}=\frac{-\zeta-2+2\sqrt{\zeta+1}}{2\zeta^2\sqrt{\zeta+1}}=\frac{-(\sqrt{\zeta+1}-1)^2}{2\zeta^2\sqrt{\zeta+1}}.
	\end{equation}
	From (\ref{eq:tiled_gamma_diff}), we find that $ \diff{\tilde{\gamma}_{1,c}}{\zeta} <0$ for $\zeta>0$, and thus $ \tilde{\gamma}_{1,c} $ in (\ref{eq:boundGamma1}) is a decreasing function.  This implies that if $ 0<\zeta_{\min}\le \zeta \le\zeta_{\max}$, then $ \tilde{\gamma}_{1,c}(\zeta_{\max}) \le\tilde{\gamma}_{1,c} \le\tilde{\gamma}_{1,c}(\zeta_{\min})$. 
	\begin{equation}\label{eq:boundGamma1_approx}
	\left.
	\begin{aligned}
	\tilde{\gamma}_{1,c} \rightarrow \frac{-1+\sqrt(1+\zeta)}{\zeta}
	\end{aligned}
	\right\}
	\Longrightarrow
	\left\lbrace 
	\begin{aligned}
	\lim_{\zeta \to 0}\tilde{\gamma}_{1,c}= \lim_{\zeta \to 0} \frac{-1+(1+\frac{1}{2}\zeta)}{\zeta}=\frac{1}{2}\\
	\lim_{\zeta \to \infty}\tilde{\gamma}_{1,c} = \lim_{\zeta \to \infty}\frac{1}{1+\sqrt{1+\zeta}} = 0 
	\end{aligned} 
	\right\}
	\Longrightarrow
	\left\lbrace 
	\begin{aligned}
	0<\gamma_{1,c} < \frac{1}{2} 
	\end{aligned}
	\right. 	
	\end{equation}	
	Accordingly, based on (\ref{eq:boundGamma1_approx}), $ \tilde{\gamma}_{1,c}  $ lies on $ (0,\frac{1}{2}) $ for $\zeta>0$. Here, Lemma \ref{lem:boundGamma1} is proved.
\end{proof}

\begin{remark}
When $ R_c $ is a slowly-increasing or slowly-decreasing function, the absolute difference between $ R_c $ at $ \gamma_{1,c}=0 $ and $ R_c $ at $ \gamma_{1,c}=1 $ is very small, i.e., $|R_c(1)-R_c(0)|= |\log_2(1+\zeta_{1,c})-\log_2(1+\zeta_{2,c})|< \epsilon $. However, since $\gamma_{1,c} \in \Gamma_c =[0,\hat{\gamma}_c] \subset [0,1] $, then $|\log_2(1+\zeta_{1,c})-\log_2(1+\zeta_{2,c})|< \epsilon$ is also valid in the feasible interval $\Gamma_c $ of $ \gamma_{1,c} $, i.e., $ |R_c(\hat{\gamma}_c )-R_c(0)|< \epsilon $. Therefore, we use $|\log_2(1+\zeta_{1,c})-\log_2(1+\zeta_{2,c})|< \epsilon$ instead of $|R_c(\hat{\gamma}_c)-R_c(0)|< \epsilon  $, which is more complicated to compute. Moreover, the value of $\epsilon$ depends on the value of $ \log_2(1+\zeta_{l,c}) $. To this end, we use the relative difference of the corresponding data rates $ \left| \frac{\log_2(1+\zeta_{1,c})-\log_2(1+\zeta_{2,c})}{\log_2(1+\zeta_{1,c})}\right| < \epsilon $ instead of $ |\zeta_{1,c}-\zeta_{2,c}|<\epsilon $. 
\end{remark}
\section{Partial CSI-based Power Allocation Scheme}
\label{sec:PartialCSI}
Using classical DBS, UEs only feedback their estimated angles, i.e.,  $ \vec{\Theta}_{1,k}=(\theta_{1,k},\phi_{1,k}) $ $ (\forall k \in \mathcal{K}) $. Accordingly, we design a NOMA-DBS scheme that uses these angles to schedule UEs between DBS and NOMA-DBS clusters. Subsequently, in Section \ref{sec:FullCSI}, we propose an intra-cluster PA scheme for NOMA-DBS, denoted as OPA, that maximizes the system throughput under full CSI. This implies that NOMA-DBS misses the key features of DBS, namely partial CSI and low channel overhead. To this end, we will formulate an alternative problem from $ P2^c $, by defining a partial CSI-based function, to guarantee low feedback overhead. First, we will try to rewrite $ \text{SINR}_{l,c}^{\text{NOMA-DBS}} $ of UE$_{l,c}$ served in NOMA-DBS cluster, by exploiting the impacts of the NLOS and LOS paths on both received power and interference at UE$_{l,c}$. Then, we will rewrite the system throughput to conclude a partial CSI-based geometric function. 

Based on (\ref{eq:channelvector}), $ |\textbf{h}^H_{l,c}\textbf{a}_c|^2 $ can be expressed as:

\begin{subequations}\label{eq:channelsteering}
	\begin{align}
	|\textbf{h}^H_{l,c}\textbf{a}_c|^2=&\bigg|\sum_{n=1}^{N_{l,c}} \alpha_{n,l,c}\textbf{a}^H_{n,l,c}\textbf{a}_c\bigg|^2
	\end{align}
	\begin{align}
	\label{eq:channelsteering1}
	= &\bigg|\alpha_{1,l,c}\bigg( \textbf{a}^H_{1,l,c}\textbf{a}_c+\sum_{n=2}^{N_{l,c}} \frac{\alpha_{n,l,c}}{\alpha_{1,l,c}}\textbf{a}^H_{n,l,c}\textbf{a}_c\bigg) \bigg|^2.
	\end{align}
\end{subequations}
In (\ref{eq:channelsteering1}), the first term represents the spatial interference of the $ c $th cluster with the LOS path in UE$_{l,c}$'s channel. And, the second term represents the spatial interference of the $ c $th cluster, with all the NLOS paths in UE$_{l,c}$'s channel.

In high signal-to-noise ratio (SNR) regime and using (\ref{eq:channelsteering}), $ R_{1,c} $ in (\ref{eq:sinr_noma_1}) and $ R_{2,c}  $ in (\ref{eq:sinr_noma_2}) are given by (\ref{eq:R1c_highSNR}) and (\ref{eq:R2c_highSNR}), respectively:
	\begin{equation}\label{eq:R1c_highSNR}
R_{1,c} \overset{\text{high SNR}}{\longrightarrow} B\log_2\bigg(1+ \frac{\gamma_{1,c}p_c| \textbf{a}^H_{1,1,c}\textbf{a}_c+\sum_{n=2}^{N_{1,c}} \frac{\alpha_{n,1,c}}{\alpha_{1,1,c}}\textbf{a}^H_{n,1,c}\textbf{a}_c|^2}{\mathlarger{\sum}_{b \in \mathcal{C}, b\neq c } p_b|\textbf{a}^H_{1,1,c}\textbf{a}_b+\sum_{n=2}^{N_{1,c}} \frac{\alpha_{n,1,c}}{\alpha_{1,1,c}}\textbf{a}^H_{n,1,c}\textbf{a}_{b}|^2}\bigg)
\end{equation}
\begin{equation}
\label{eq:R2c_highSNR}
R_{2,c} \overset{\text{high SNR}}{\longrightarrow} B\log_2\bigg(1+ \frac{\gamma_{2,c}p_c\left|  \textbf{a}^H_{1,2,c}\textbf{a}_c+\sum_{n=2}^{N_{2,c}} \frac{\alpha_{n,2,c}}{\alpha_{1,2,c}}\textbf{a}^H_{n,2,c}\textbf{a}_c\right| ^2}{\mathlarger{\gamma_{1,c}p_c\left|  \textbf{a}^H_{1,2,c}\textbf{a}_c+\sum_{n=2}^{N_{2,c}} \frac{\alpha_{n,2,c}}{\alpha_{1,2,c}}\textbf{a}^H_{n,2,c}\textbf{a}_c\right| ^2+\sum}_{b \in \mathcal{C}, b\neq c} p_b\left| \textbf{a}^H_{1,2,c}\textbf{a}_b+\sum_{n=2}^{N_{2,c}} \frac{\alpha_{n,2,c}}{\alpha_{1,2,c}}\textbf{a}^H_{n,2,c}\textbf{a}_{b}\right| ^2}\bigg)
\end{equation}

As seen in (\ref{eq:R1c_highSNR}) and (\ref{eq:R2c_highSNR}), the NLOS paths affect either the received power or the interference. Indeed, the interference $I_{:\rightarrow l,c},$  $ (l={1,2})$ is due to the beams of other UEs that interferes with LOS and NLOS paths. In addition, the received power at UE$_{l,c}$ is that transmitted by the beam interfered with the LOS path plus the effect of paths in the vicinity of the beam direction. 

Assuming that this NLOS paths' impact is negligible, i.e., $ \frac{\alpha_{n,l,c}}{\alpha_{1,l,c}} $ is negligible, thus the data rates $\breve{R}_{l,c}$ at UE$_{l,c}$ $ (l={1,2})$ can be approximated as follows:
\begin{equation}\label{eq:R1c_highSNR_approx}
\breve{R}_{1,c}= B\log_2\bigg(1+ \frac{\gamma_{1,c}p_c| \textbf{a}^H_{1,1,c}\textbf{a}_c|^2}{\mathlarger{\sum}_{b \in \mathcal{C}, b\neq c} p_b|\textbf{a}^H_{1,1,c}\textbf{a}_b|^2}\bigg),
\end{equation}
\begin{equation}
\label{eq:R2c_highSNR_approx}
\breve{R}_{2,c} = B\log_2\bigg(1+ \frac{\gamma_{2,c}p_c| \textbf{a}^H_{1,2,c}\textbf{a}_c|^2}{\mathlarger{\gamma_{1,c}p_c| \textbf{a}^H_{1,2,c}\textbf{a}_c|^2+\sum}_{b \in \mathcal{C}, b\neq c} p_b|\textbf{a}^H_{1,2,c}\textbf{a}_b|^2}\bigg).
\end{equation}
As seen in (\ref{eq:R1c_highSNR_approx}) and (\ref{eq:R2c_highSNR_approx}), $\breve{R}_{l,c}$ depends on the LOS path angle information, i.e., $ \vec{\Theta}_{1,l,c} $, the $ \beta $-UC, i.e., $\vec{\Theta}_c$, and the applied PA method, i.e., $\gamma_{1,c}$ and $\gamma_{2,c}$.
The obtained partial CSI-based function $\breve{R}_{l,c}$ corresponds to the data rate at UE$ _{l,c} $ in case of mono-path environment, where the received power just comes from the LOS path and the interference comes from the other-clusters with the LOS path. However, DBS is a LOS path angle-based beamformer that generates a single beam toward each UE. Thus, in multi-path environment, DBS loses the potentiality of the NLOS paths and some of the radiated energy, as obtained in \cite{Israakhaled2020WSA}. To guarantee the key feature of DBS, we formulate a partial CSI-based function that misses the potentiality of the NLOS paths. Accordingly, the problem $ P3^c $ from $ P2^c $ consists to find the PA coefficients that maximize the partial CSI function $\breve{R}_{c}$ instead of the NOMA-DBS cluster's throughput $ {R}_{c} $, as follows: 
\begin{equation}\label{pb:OptimizedPA3} 
\begin{aligned} 
P3^c: \quad \stackrel{\star}{\gamma}_{1,c}= &\underset{\gamma_{1,c} }{\max} \ \breve{R}_c= \underset{\gamma_{1,c} }{\max} \ \breve{R}_{1,c}+\breve{R}_{2,c}; \ \forall c \in \mathcal{C}_n,\\
\text{s.t}:&\ C_1:{\gamma_{1,c}}+{\gamma_{2,c}}=1 \ \forall c \in \mathcal{C}_n, \\
&\ C_2':\gamma_{1,c}\le\frac{1}{2},\\
&\ C_3:\gamma_{2,c}-\gamma_{1,c}\ge  \frac{P_{\min}}{\breve{\zeta}_{1,c}}.
\end{aligned}
\end{equation}

This Problem $ P3^c$ will be solved in the same manner as Problem $ P2^c $, but with $\breve{\zeta}_{l,c}=\frac{\breve{\psi}_{l,c}}{\breve{\nu}_{l,c}}$, where:
\begin{equation}\label{}
\breve{\psi}_{l,c}=\eta p_c|\textbf{a}^H_{1,l,c}\textbf{a}_c|^2, \quad l \in \mathcal{K}_c, c \in \mathcal{C}_n,
\end{equation}
\begin{equation}\label{}
\breve{\nu}_{l,c}=\mathlarger{\sum}_{b \in \mathcal{C}, b\neq c }\eta p_b|\textbf{a}^H_{1,l,c}\textbf{a}_{b}|^2, \quad l \in \mathcal{K}_c, c \in \mathcal{C}_n.
\end{equation}
\section{Performance Evaluation}
\label{sec:simulationresults}

In this section, numerical results are presented to verify the performance of the proposed NOMA-DBS scheme. All results are obtained by averaging over $ 5000 $ random trials. Perfect estimation of angles at the receivers and perfect instantaneous feedback of the estimated angles are assumed in the analysis of such systems. We set $ \beta_0 =0.5 $ as the spatial interference threshold. Other system and channel parameters used for performance assessment are listed in Table \ref{tab:sumparam}. 
\begin{table}[!h]
	\centering
	\caption{ Simulation parameters}
	\begin{tabular}{ p{8.75cm}|p{2.25cm}}
		\hline 
		\textbf{System and Channel Parameters} & \textbf{Values}\\
		\hline  
		Number of transmit antennas & 64\\
		Carrier frequency & 28 GHz\\
		Channel bandwidth & 20 MHz\\
		Cell edge radius & 100 m\\
		Transmission power & 30 dBm\\
		Noise power & -100.9178 dBm\\
		Minimum power difference $P_{\min}$ & 1 mW\\
		Number of paths per time cluster in rural scenario & $ \{1,2\} $\\
		\hline   
	\end{tabular} 
	\label{tab:sumparam}
\end{table}

The proposed schemes with full and  partial CSI are labeled, respectively, as "$\beta$-UC NOMA-DBS-FCSI" and "$\beta$-UC NOMA-DBS-PCSI". For comparison, the orthogonal multiple access (OMA) with DBS (OMA-DBS) is adopted as the baseline, in which the two UEs in NOMA-DBS cluster obtained using $\beta$-UC algorithm are served using OMA with equal degrees of freedom for each UE, which is labeled as "$\beta$-UC OMA-DBS". Moreover, a widely used beamformer, namely maximum rate ratio, also known as conjugate beamforming (CB) \cite{Yang2013} is adopted to make a comparison.
The number of real coefficients needed per channel estimation for the different systems is summarized in Table \ref{tab:channelFeedback}.
\begin{table}[!h]
	\centering
	\caption{ Channel Overhead}
	\begin{tabular}{ p{5.5cm}|p{7.5cm}}
		\hline 
		\textbf{System } & \textbf{Number of real coefficients per channel estimation for each UE}\\
		\hline  
		Classical DBS  & 2 $  [\theta_{1,k},\phi_{1,k}] $\\
		MRT & $ 2\times M $ [$ \textbf{H} $] \\
		$\beta$-UC NOMA-DBS-FCSI & $ 2\times (M+1) $   [$ \textbf{H}+(\theta_{1,k},\phi_{1,k}) $]\\
		$\beta$-UC NOMA-DBS-PCSI & 3 [$ |\textbf{h}_{l,c}\textbf{w}_{c}|^2$ + $ (\theta_{1,k},\phi_{1,k}) $] \\
		$\beta$-UC OMA-DB & $ 2\times M $ [$ \textbf{H} $]\\
		\hline   
	\end{tabular} 
	\label{tab:channelFeedback}
\end{table} 
\subsection{Spectral Efficiency}
Fig. \ref{fig:sumrate_UPA_RMa} depicts the spectral  efficiency (or the achievable sum-rate per the normalized bandwidth) versus the total number of users, in rural environment, when BS is equipped using UPA with $ M_H=32 $.
\begin{figure}[h!]
	\centering
	\includegraphics[scale=1]{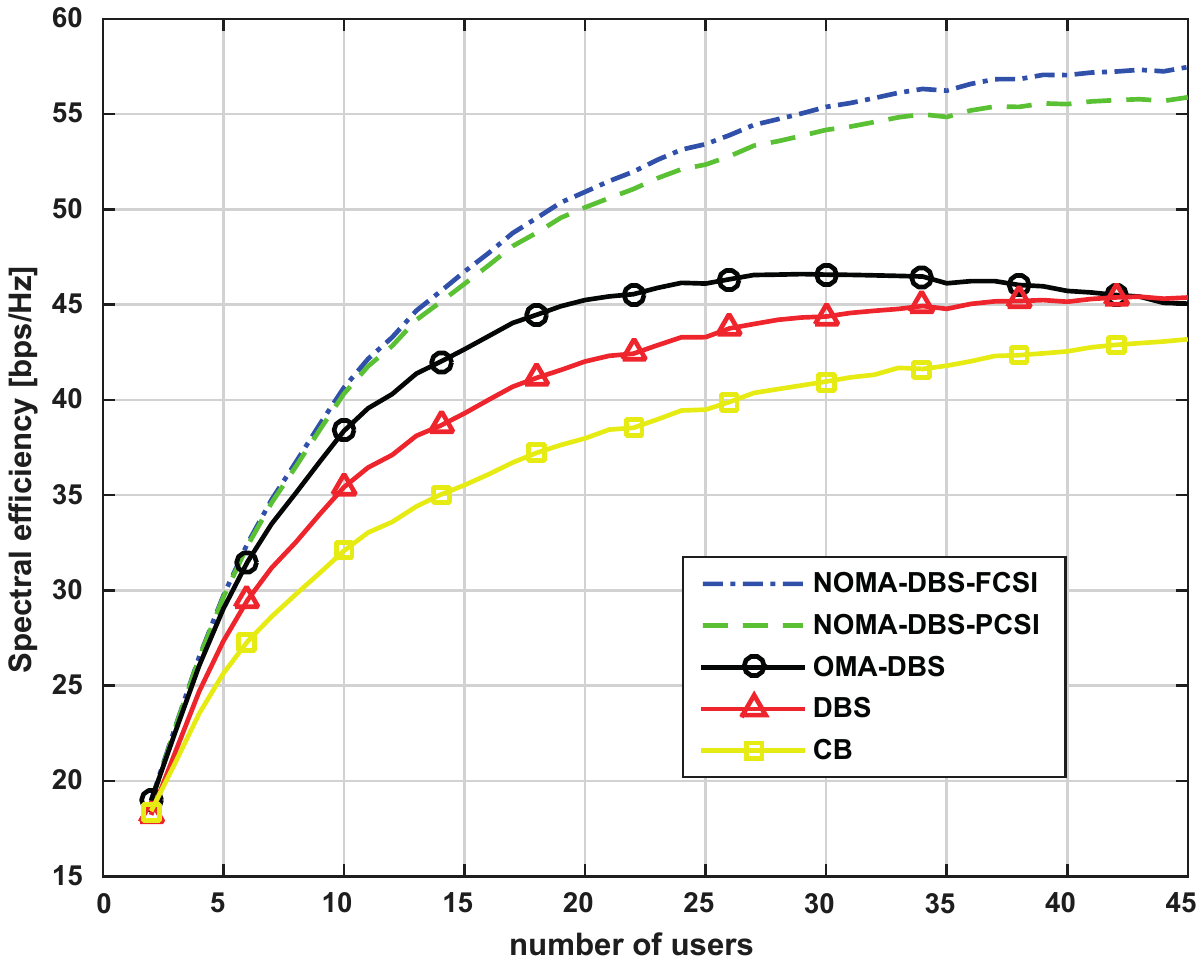}
	\caption{Spectral  efficiency versus number of users, for NOMA-DBS, OMA-DBS, classical DBS and CB in rural environment, when BS is equipped with UPA.}
	\label{fig:sumrate_UPA_RMa}
\end{figure}

It can be seen that the proposed NOMA-DBS scheme with full and partial CSI significantly improves the performance of classical DBS. Fortunately, attractive results are achieved, where NOMA-DBS-PCSI can almost achieve the same performance as  NOMA-DBS-FCSI, with much lower channel feedback; $ 43.3 $ times less as indicated in Table \ref{tab:channelFeedback}. Indeed, one or two paths exist in rural environment (see Table \ref{tab:sumparam}), and thus the partial CSI-based function in NOMA-DBS-PCSI scheme is approximately the same as the throughput. For instance, NOMA-DBS-PCSI achieves higher spectral  efficiency with respect to DBS up to $ 26.8 \% $ for 45 users in the cell using UPA. 

As seen in Fig. \ref{fig:sumrate_UPA_RMa}, it is obvious that classical DBS outperforms CB in rural environment, thanks to a few paths as explained in our previous work \cite{Israakhaled2020WSA}. Besides, NOMA-DBS improves the performance of DBS and thus surpasses CB. For instance, the performance gain in spectral  efficiency for NOMA-DBS-PCSI against CB is approximately $ 29.4\% $ for 45 users in the cell using UPA, with lower complexity and much lower channel overhead ($ 42.6 $ times less as shown in Table \ref{tab:channelFeedback}). 

It can be seen from Fig. \ref{fig:sumrate_UPA_RMa} that NOMA-DBS surpasses the   conventional OMA-DBS schemes in rural environment, and using UPA. This is also obtained in the prior work comparing NOMA-MIMO with OMA-MIMO, which concludes that OMA-MIMO can not support more users, while NOMA-MIMO enables a higher number of successfully connected users. Indeed, NOMA with the aid of SIC and SC enables more UEs in each beam to be supported at the same spectrum efficiency, while OMA serves only one UE in each beam.  
\subsection{Energy Efficiency}
The energy efficiency ($ \text{EE} $) is also a key performance metric used for future cellular network, because of environmental and economic interests. $\text{EE}$ is defined as the ratio of the spectral efficiency $ [bps] $ to the total power consumed at BS:
\begin{equation}\label{}
\text{EE}=\frac{B\sum_{k=1}^{K}R_k}{\rho \sum_{k=1}^{K}\lVert \mathbf{w}_k \rVert^2+MP_a+P_0} \quad [bps/J].
\end{equation}
where $ \rho \ge 1 $ is a constant modeling the inefficiency of the power amplifier, $ P_a $ is the constant power consumption per antenna independent of the transmitted power, and $P_0 $ is the basic power consumed at the base station independent of the number of antennas. Assume $ \rho=10 $, $ P_a=1W $ and $P_0=200mW  $.
\begin{figure}[h!]
	\centering
	\includegraphics[scale=1]{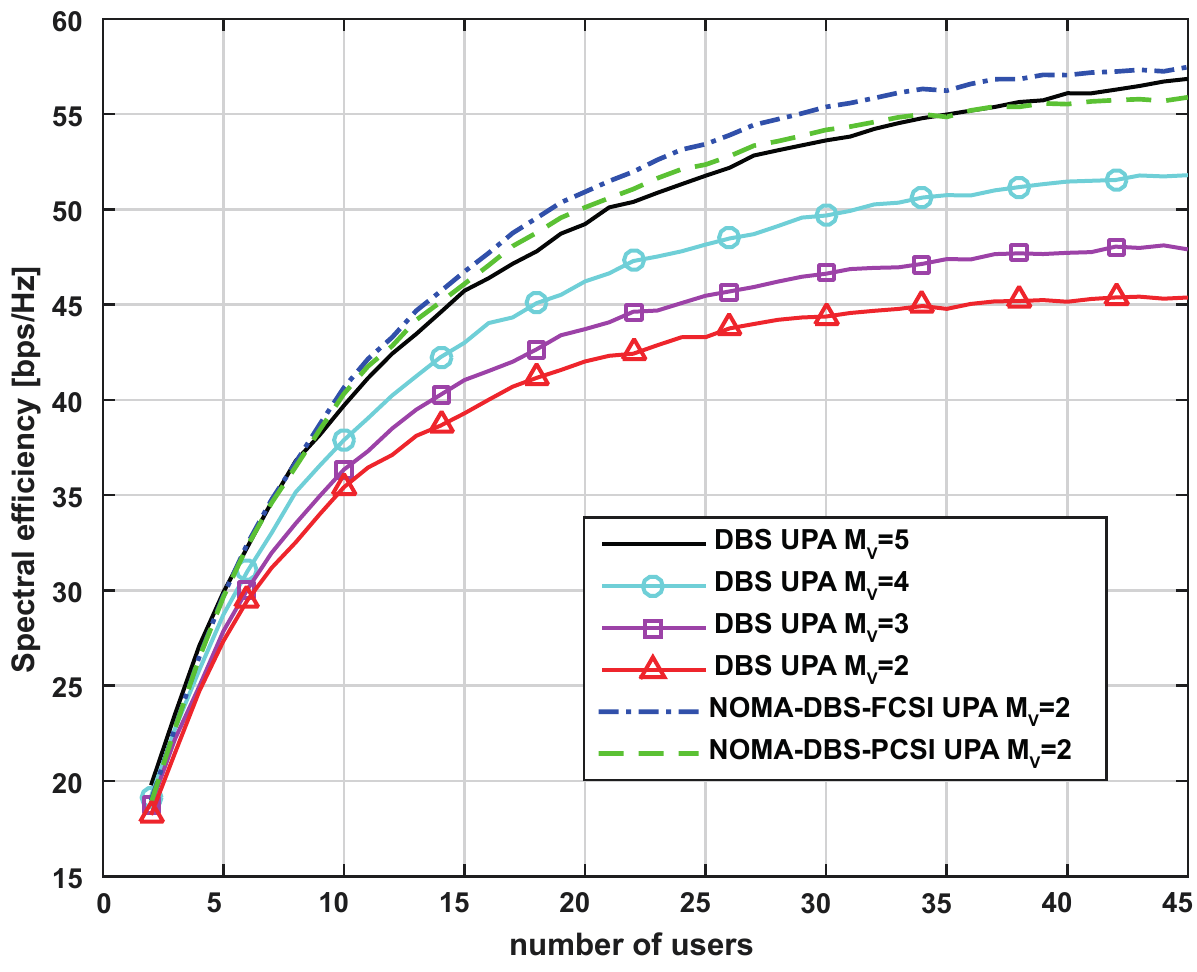}
	\caption{Spectral  efficiency versus number of users, for NOMA-DBS using UPA with $ M_H=32 $ and $ M_V=2 $, and for classical DBS using UPA with $ M_H=32 $ and different values of $ M_V $ ($ M_V=\{2,3,4,5 \}$) in rural environment.}
	\label{fig:sumrate_UPA32_RMa}
\end{figure} 
\begin{figure}[h!]
	\centering
	\includegraphics[scale=1]{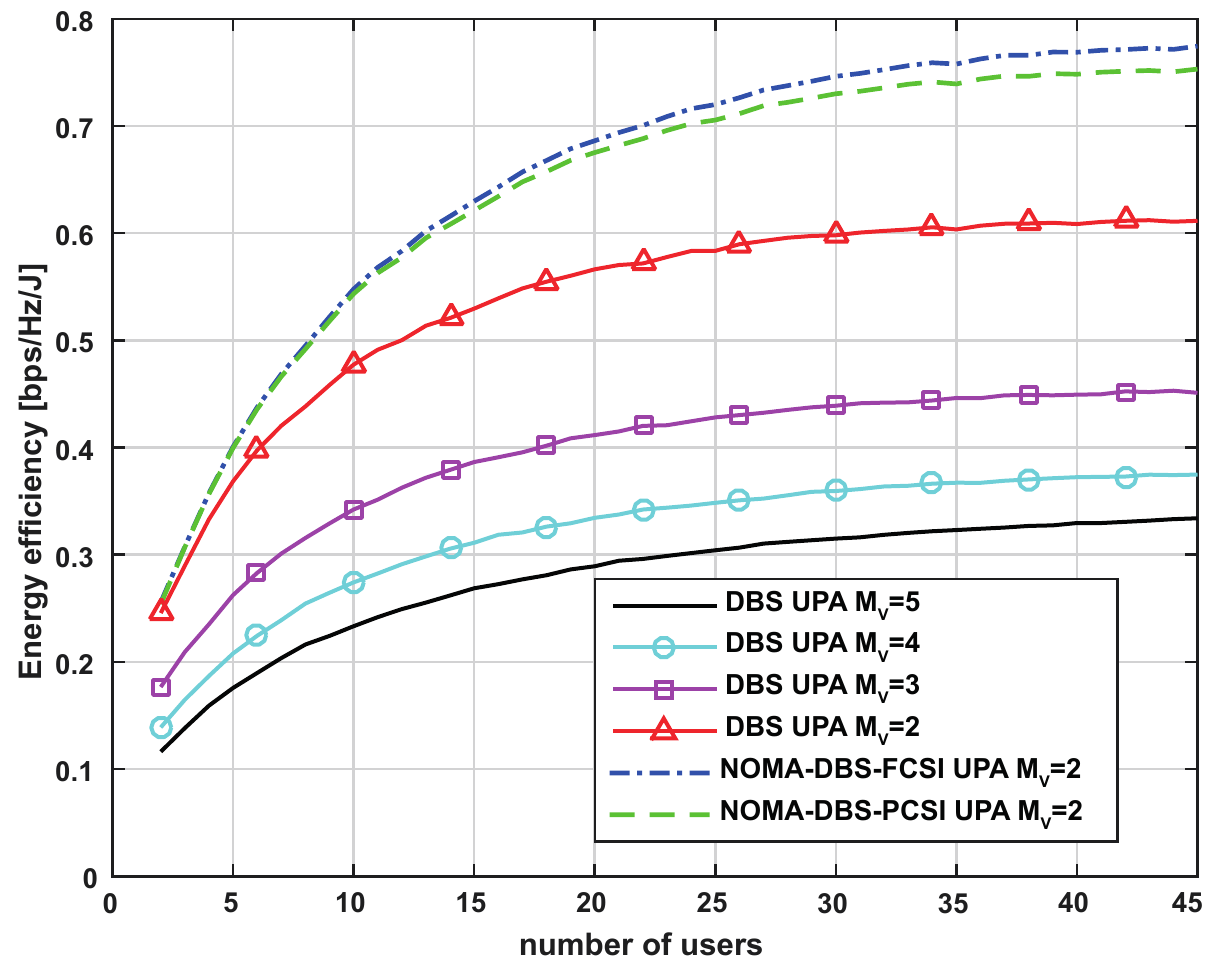}
	\caption{Energy efficiency versus number of users, for NOMA-DBS using UPA with $ M_H=32 $ and $ M_V=2 $, and for classical DBS using UPA with $ M_H=32 $ and different values of $ M_V $ ($ M_V=\{2,3,4,5 \}$) in rural environment.}
	\label{fig:EE_UPA32_RMa}
\end{figure}

Fig. \ref{fig:sumrate_UPA32_RMa} and \ref{fig:EE_UPA32_RMa} depict, respectively, the spectral  efficiency and energy efficiency versus number of users, for NOMA-DBS using UPA with $ M_H=32 $ and $ M_V=2 $, and for classical DBS using UPA with $ M_H=32 $ and different values of $ M_V $ ($ M_V=\{2,3,4,5 \}$) in rural environment. 

In rural environment, using $ 96 $ less of antennas, our proposed schemes NOMA-DBS-FCSI and NOMA-DBS-PCSI with $ M_H=32 $ and $ M_V=2 $ have almost the same performance of the classical DBS with $ M_H=32 $ and $ M_V=5 $ in terms of spectrum efficiency, as shown in Fig. \ref{fig:sumrate_UPA32_RMa}. It's obvious therefore that the energy efficiency using NOMA-DBS is $ 2.2 $ times greater than that using classical DBS with $ M_H=32 $ and $ M_V=5 $, as seen in Fig. \ref{fig:EE_UPA32_RMa}. Accordingly, NOMA-DBS is an alternative solution to enhance the DBS performance without increasing the number of antennas.
\section{Conclusion}  
\label{sec:Conclusion}
In this paper, we integrate DBS with PD-NOMA transmission for mmWave massive MIMO systems in order to boost the DBS performance without increasing the number of antennas. We first define a geometric interference metric, so that the proposed NOMA-DBS scheme regrouped the users based on their directions. Subsequently, by applying a fixed inter-cluster PA, we derive an intra-cluster PA for 2-user NOMA-DBS system that maximizes the sum-throughput. More importantly, a new partial CSI-based geometric function is defined by exploiting the spatial behavior of DBS, to design a PA scheme with low-feedback rate. Simulation results using NYUSIM indicates that our proposed low-complex NOMA-DBS system with limited feedback yields a significant spectral efficiency improvements with respect to DBS, and can be an alternative solution to the use of a large number of antennas. However, our proposed scheme is limited to 2-user in NOMA-DBS cluster. It is in our interest to extend 2-user NOMA-DBS to support multiple users, to further enhance the NOMA-DBS scheme. 

\section{Appendix}
\label{ap:derivative_Rc}
\subsection{Appendix 1}
The first-order derivative of $ R_{1,c} $, $ R_{2,c} $ and $R_c$ can be expressed as:
\begin{equation}\label{}
\diff{R_{1,c}}{\gamma_{1,c}}=\ln(2)\frac{\psi_{1,c}}{\nu_{1,c}+\psi_{1,c}\gamma_{1,c}}>0,
\end{equation}
\begin{equation}\label{}
\diff{R_{2,c}}{\gamma_{1,c}}=-\ln(2)\frac{\psi_{2,c}}{\nu_{2,c}+\psi_{2,c}\gamma_{1,c}}<0,
\end{equation}
\begin{subequations} \label{eq:Rc_derivative}
	\begin{align}
	\diff{R_{c}}{\gamma_{1,c}}&=\diff{R_{1,c}}{\gamma_{1,c}}+\diff{R_{2,c}}{\gamma_{1,c}}
	\end{align}	
	\begin{align}
	&=\ln(2)\frac{\zeta_{1,c}-\zeta_{2,c}}{(1+\zeta_{1,c}\gamma_{1,c})(1+\zeta_{2,c}\gamma_{1,c})}.
	\end{align}		
\end{subequations}
(\ref{eq:Rc_derivative}) indicates that the sign of the first-order derivative of  $R_c$ is the same of that of $\zeta_{1,c}-\zeta_{2,c}$. Here, Lemma \ref{lem:RcMonotone} is proved.

\bibliographystyle{iet}
\bibliography{Bibliography}	
\end{document}